\documentclass[12pt]{article}
\usepackage{amsmath,amsthm,amssymb,amscd,fancybox,epsfig,float}
\usepackage[margin=1in]{geometry}
\usepackage{graphics}
\usepackage{epstopdf}
\usepackage{multirow}
\usepackage{hhline}
\usepackage{cleveref}
\usepackage{xcolor}
\usepackage{tikz}
\usetikzlibrary{arrows,automata}
\usepackage[nocompress]{cite}
\usepackage{pgfplots}
\usepackage{subcaption}
\usepackage{blindtext}

\newcommand{\beqn}{\begin{equation}}
\newcommand{\eeqn}{\end{equation}}

\newcommand{\cN}{\mathcal{N}}

\newcommand{\hk}{\hat{k}}

\newcommand{\R}{\mathbb{R}}

\usepackage{bm}
\newcommand\vct[1]{\bm{\mathsf{#1}}}
\newcommand\mtx[1]{\bm{\mathsf{#1}}}

\newcommand{\tran}{^{\mathsf{T}}}

\newtheorem{theorem}{Theorem}

\newtheorem{algorithm}{Algorithm}

\title{Efficient Fourier representations of families of Gaussian processes}
\author{Philip Greengard\thanks{Department of Statistics, Columbia University, pg2118@columbia.edu. Research supported by the Alfred P. Sloan Foundation.}}
\date{May 31, 2024}
\begin{document}

\maketitle
%%%\tableofcontents 
%%%\newpage 

\begin{abstract}
We introduce a class of algorithms for constructing Fourier representations of Gaussian processes in $1$ dimension that are valid over ranges of hyperparameter values. The scaling and frequencies of the Fourier basis functions are evaluated numerically via generalized quadratures. The representations introduced allow for $O(m^3)$ inference, independent of $N$, for all hyperparameters in the user-specified range after $O(N + m^2\log{m})$ precomputation where $N$, the number of data points, is usually significantly larger than $m$, the number of basis functions. Inference independent of $N$ for various hyperparameters is facilitated by generalized quadratures, and the $O(N + m^2\log{m})$ precomputation is achieved with the non-uniform FFT.  Numerical results are provided for Mat\'ern kernels with $\nu \in [3/2, 7/2]$ and lengthscale $\rho \in [0.1, 0.5]$ and squared-exponential kernels with lengthscale $\rho \in [0.1, 0.5]$. The algorithms of this paper generalize mathematically to higher dimensions, though they suffer from the standard curse of dimensionality.
\end{abstract}

\section{Introduction}
Gaussian process (GP) regression has become ubiquitous as a statistical tool in many applied sciences including astrophysics, environmental sciences, molecular dynamics, financial economics, and social sciences~\cite{banerjee1, dfm1, bda, gonzalvez1, stein2, cressie1, bartok1, luger1}. 

In most Gaussian process regression problems a user has observed data $\{(x_i, y_i)\}$ where $x_1,...,x_N$ are independent variables that belong to some set $D \subseteq \R^d$ and $y_i \in \R$ are dependent variables. The observed data $y_1,...,y_N$ are assumed to be observations of the form 
\begin{align}
y_i = f(x_i) + \epsilon_i
\end{align}
where $\epsilon_i$ is independent and identically distributed (iid) Gaussian noise, and $f: D \to \R$ is an unknown function, which is given a Gaussian process prior $f(x) \sim \mathcal{GP}(\mu(x), k(x,x'))$ where $\mu$ is a user-specified mean function and $k$ is a covariance function~\cite{rasmus1}. 

The primary computational limitation of GP regression as a practical tool is the  
cost of the matrix inversion and determinant that appear in the likelihood function 
of a Gaussian process
\begin{align}\label{63}
p(\vct{y}) \propto \frac{1}{|\mtx{K} + \sigma^2 \mtx{I} |^{1/2}} \exp^{-\frac{1}{2} \vct{y}\tran (\mtx{K} + \sigma^2\mtx{I})^{-1} \vct{y}},
\end{align}
where $\mtx{K}$ is the $N \times N$ matrix such that $\mtx{K}_{i,j} = k(x_i, x_j)$ and $\sigma^2$ is the variance of the iid observation noise. For a general symmetric matrix, direct inversion and determinant evaluation both require $O(N^3)$ operations, which for many modern problems is far too costly. A large body of literature has emerged over the last couple of decades on efficient schemes for evaluating these quantities when $N$ is large (i.e. greater than around 10,000). Often, these computational methods rely on taking advantage of the particular structure of covariance matrices~\cite{oneil1, minden1, dfm1}, approximating the covariance matrix with a low-rank matrix (reduced-rank methods)~\cite{candela1, solin1, greengard2}, spectral methods~\cite{hensman1, lazaro2010, rahimi2007}, or many more strategies. 

In this paper we introduce numerical algorithms for representing a zero-mean Gaussian process $f$ with a translation-invariant covariance kernel $k$ as a random expansion of the form
\begin{align}\label{23}
f(x) \sim \sum_{i=1}^{m} \alpha_i \gamma_i \cos(2 \pi \xi_i x) + \beta_i \gamma_i \sin(2 \pi \xi_i x),
%+ ... + \alpha_m \gamma_m \cos(2 \pi \xi_n x) + \beta_m \gamma_m \sin(2 \pi \xi_m x),
\end{align}
where for $i=1,...,m$ the frequencies $\xi_i \in \R$ are fixed, and $\alpha_i$ and $\beta_i$ are the iid Gaussians
\begin{align}
\alpha_i, \beta_i \sim \cN \big(0, 1 \big).
\end{align}
The coefficients $\gamma_i$ of (\ref{23}) are defined by 
\begin{align}\label{gamma_i}
\gamma_i = \sqrt{2 w_i \hk(\xi_i)},
\end{align}
where $\hk$ denotes the Fourier transform of the covariance kernel and $w_i > 0$. The numerical work in constructing expansion (\ref{23}) involves  finding the frequencies $\xi_i$ and the weights $w_i$ such that expansion (\ref{23}) has an effective covariance kernel that approximates the desired kernel (or family of kernels) to high accuracy. We provide numerical results for constructing an expansion of the form (\ref{23}) that is valid for all Mat\'ern kernels with $\rho \in [0.1, 0.5]$, and $\nu \in [3/2, 7/2]$. We also include results for approximating \eqref{23} with squared-exponential kernels with $\rho \in [0.1, 0.5]$. 

There are several advantages to using expansion (\ref{23}) as a representation of a family of GPs. One computational benefit stems from the so-called weight-space approach to Gaussian process regression (see e.g. \cite{rasmus1}, \cite{greengard2}, \cite{filip2019random}) in which the posterior density is defined over the coefficients of a basis function expansion such as \eqref{23}. Specifically, given data $\{(x_i, y_i)\}$ and a covariance kernel $k$, weight-space Gaussian process regression is the $L^2$-regularized linear regression 
\begin{equation}\label{72}
\begin{aligned}
\vct{y} & \sim \mtx{X}\vct{\beta} + \vct{\epsilon} \\ % \sum_{j=1}^k (\alpha_j\cos(\xi_j x_i) + \beta_j \sin(\xi_j x_i))  \\
\vct{\epsilon} & \sim \cN(0, \sigma^2 \mtx{I}) \\
\vct{\beta} & \sim \cN(0, \mtx{I})
\end{aligned}
\end{equation}
where $\vct{\beta} \in \R^{2m}$ is the vector of coefficients and $\mtx{X}$ is the $N \times 2m$ matrix 
\begin{equation}\label{x_mat}
\mtx{X} = 
\renewcommand*{\arraystretch}{1.3}
\begin{bmatrix}
    \gamma_1 \cos(\xi_1 x_1) & \dots  & \gamma_m \cos(\xi_m x_1) & \gamma_1 \sin(\xi_1 x_1) & \dots & \gamma_m \sin(\xi_m x_1) \\
    \gamma_1 \cos(\xi_1 x_2) & \dots  & \gamma_m \cos(\xi_m x_2) & \gamma_1 \sin(\xi_1 x_2) & \dots & \gamma_m \sin(\xi_m x_2) \\
    \vdots &  & \vdots & \vdots & & \vdots \\
    \gamma_1\cos(\xi_1 x_N) & \dots  & \gamma_m\cos(\xi_m x_N) & \gamma_1\sin(\xi_1 x_N) & \dots & \gamma_m\sin(\xi_m x_N)
\end{bmatrix}.
\end{equation}
The posterior density corresponding to (\ref{72}) is defined over $2m$ dimensions (assuming $\sigma$ is fixed) and the posterior mean function is given by
\begin{align}
\sum_{i=1}^{m} \vct{\bar{\beta}}_{1,i} \gamma_i \cos(\xi_i x) + \vct{\bar{\beta}}_{2,i} \gamma_i \sin(\xi_i x)
\end{align}
where $[\vct{\bar{\beta}}_1 \, \vct{\bar{\beta}}_2]\tran$ is the solution to the linear system of equations
\begin{align}\label{203}
(\mtx{X\tran X} + \sigma^2 \mtx{I}) \vct{\bar{\beta}} = \mtx{X\tran}\vct{y}.
\end{align}
The equivalence between the solution to \eqref{203} and the solution to the
so-called function space system in \eqref{63} is described in \cite{greengard3}.

Representing a GP with Fourier expansion (\ref{23}) has the advantage that GP regression via linear system (\ref{203}) can be solved in $O(N + m^3)$ operations. 
For a general $N \times 2m$ matrix, $\mtx{X}$, solving linear system (\ref{203}) requires $O(Nm^2)$ operations, which can be prohibitively expensive. However, since we use Fourier expansions, the $2m \times 2m$ matrix $\mtx{X\tran X}$ can be formed in $O(N + m^2\log{m})$ operations using a non-uniform fast Fourier transform (FFT)~\cite{dutt1, lgreengard1}. The efficient formation of $\mtx{X\tran X}$  reduces the total computational cost of solving the linear system from $O(Nm^2)$ operations to $O(N + m^3)$. In general, $m$ is sufficiently small that $m^3$ operations is easily affordable on a laptop. Furthermore, for GP problems that involve fitting hyperparameters, the cost of solving linear system \eqref{203} is $O(m^3)$ operations for all hyperparameters after a precomputation of $O(N + m^2\log{m})$ operations. 

The numerical work in finding $\xi_i$ and $w_i$ of Fourier expansion \eqref{23} reduces to constructing a quadrature rule (see, e.g. \cite{ma1}) for an inverse Fourier transform. Specifically, if $k$ is an integrable translation-invariant covariance kernel with Fourier transform $\hk$, then $k$ satisfies
\begin{align}
k(x) = \int_{\R} \hat{k}(\xi) e^{2\pi i \xi x} d\xi.
\end{align}
It turns out that expansion (\ref{23}) can be constructed by evaluating $\xi_i \in \R$ and $w_i \in \R^+$ such that $k(x)$ is well approximated by the sum
\begin{align}\label{intro_quad}
\sum_{j=1}^n w_j \hk(\xi_j) e^{2\pi i \xi_j x}
\end{align}
for a family of covariance kernels where $x$ is in some region of interest.

In this paper, we find $\xi_i$ and $w_i$ with algorithms for constructing generalized Gaussian quadratures. The theory associated with generalized Gaussian quadratures was originally introduced in 1966~\cite{karlin1} and more recently, efficient numerical algorithms have made constructing generalized Gaussian quadratures practical for many modern problems\cite{ma1, bremer1}. The use of such quadratures is now widespread in several environments including in computational physics for the solution of integral equations with singular kernels, and in the numerical solution of certain partial differential equations (e.g.~\cite{ma1, hoskins1, bremer2}). 

Fourier representation \eqref{23} is also used in \cite{greengard3} for representing
GPs, though in \cite{greengard3} the frequencies $\xi_1,...,\xi_m$ are equispaced (and lie on a Cartesian grid in higher dimensions). 
%Equivalently, in \cite{greengard3}, quadrature rule \eqref{intro_quad} is an equispaced quadrature. 
The primary purpose of that choice of quadrature 
is its benefits in $2$ and $3$ dimensions for spatial and spatio-temporal problems.
By using equispaced nodes, the weight space matrix of \eqref{203} 
has Toeplitz structure and can be applied in $O(m \log m)$ operations allowing the 
efficient use of iterative solvers. The dense linear algebra solvers that are used 
in this paper are impractical in high dimensions since the number of basis functions
required for a fixed level of accuracy grows exponentially in dimension. 

In one dimension, the generalized quadrature approach of this paper can offer
a significant advantage over the equispaced approach, especially for 
problems that involve fitting hyperparameters. 
It is common to fit hyperparameters over ranges of values that include kernels that 
concentrate near zero in the spectral domain (e.g., large lengthscale $\rho$ in Mat\'ern kernel \eqref{mat_ker}) 
as well as those that have slow decay at infinity in the spectral domain
(e.g., small $\nu$ in \eqref{mat_ker}). Using 
an equispaced quadrature rule to simultaneously integrate those kernels
 would require a small grid spacing over a large interval and thus a large 
 number of nodes. 
Here, we use algorithms for generalized quadratures to construct quadrature 
rules that simultaneously discretize such families of functions to high accuracy. 

The basis function approach of this paper is similar in spirit to several popular
basis function approaches including \cite{lazaro2010, rahimi2007, greengard2}. 
The methods of this paper, however, have two primary advantages. 

\begin{itemize}

\item 
Fourier basis functions are amenable to fast algorithms for solving the linear systems of GP regression. We capitalize on this fact to obtain $O(N + m^3)$ computational complexity for GP regression. 

\item 
The Fourier expansions of this paper are valid over families of covariance kernels, 
whereas other basis function approaches require recomputing basis functions 
for kernels with different hyperparameter values. As a result, 
adaptation of hyperparameters is simplified. For all hyperparameter values in the domain of the Fourier expansion, GP regression is performed in $O(m^3)$ operations after a precomputation of $O(N + m^3)$ operations. 

\end{itemize}

There have been a number of kernel approximation methods introduced 
that, like the methods of this paper, approximate a kernel via quadrature 
in the Fourier domain of a translation-invariant kernel. For example, 
the random Fourier features approach of \cite{rahimi2007} 
uses a Monte-Carlo integration in Fourier domain and \cite{munkhoeva2018}
introduces a more general set of quadratures that also rely on random features. 
Deterministic, higher-order quadrature methods have also been introduced 
for this purpose. For example, in \cite{dao2017, shustin1}, Gaussian
(or Gauss-Legendre) nodes and weights are used for integration in Fourier 
domain. 
Such methods can be effective for an individual kernel, but are generally less 
useful for integrating commonly-used families of spectral densities.

The remainder of this paper is structured as follows. In the following section, we introduce theoretical and numerical tools for representing GPs as Fourier expansions. In Section \ref{s20} we describe fast algorithms for GP regression. We provide the numerical results of the algorithms of this paper in Section \ref{s120}. We conclude with a brief discussion of future directions of research in Section \ref{s125}.

\section{Spectral representation of GPs}\label{s15}
Translation invariant (or stationary) covariance kernels are commonly used in practice and include kernel families such as Mat\'ern, squared-exponential, rational quadratic, periodic, and many more~\cite{rasmus1}. A translation invariant kernel is a function $k(x, y) : \R^d \times \R^d \rightarrow \R$ that can be expressed as a function of $x - y$. In a slight abuse of notation we refer to translation invariant kernels $k(x, y)$ as functions of one variable, $k(x)$ for $x \in \R$. 

All of the previously mentioned stationary kernels are also isotropic, that is, they are functions 
$k(x, y)$ that depend on $\| x - y \|$. For those kernels, $k(x) = k(-x)$, their Fourier transforms are real-valued and symmetric. However, not all symmetric functions are valid covariance kernels. Translation invariant kernels have a particular property---$k$ is a valid kernel if and only if its Fourier transform is non-negatively valued. This property is known as Bochner's theorem~\cite{rasmus1}. 

Clearly an integrable translation-invariant kernel $k$ can be expressed as the inverse Fourier transform 
\begin{align}
k(x) = \int_{-\infty}^{\infty} \hk(\xi) e^{2\pi i\xi x} d\xi,
\end{align}
or equivalently, since $\hk$ is even, \begin{align}\label{107}
k(x) = \int_{0}^{\infty} 2 \hk(\xi) \cos(2\pi \xi x) d\xi.
\end{align}
A discretized version of integral (\ref{107}), or an order-$n$ quadrature rule for approximating the integral is a sum of the form
\begin{align}\label{110}
k(x) \approx \sum_{j=1}^m 2 w_j \hk(\xi_j) \cos(2\pi \xi_j x) =: k'(x),
\end{align}
where $\xi_j > 0$ and $w_j > 0$\footnote{The weights $w_j$ need not be
positively valued, though throughout this paper, and in general in the 
literature, $w_j$ are assumed to be positively valued.}. 

The following theorem demonstrates that any sum of the form (\ref{110}) corresponds to a basis function representation of a Gaussian process with a covariance kernel that approximates $k$ with the accuracy of discretization (\ref{110}). 

\begin{theorem}\label{120}
Let $f$ be the random expansion defined by
\begin{align}\label{56}
f(x) \sim \sum_{i=1}^{m} \alpha_i \gamma_i \cos(2 \pi \xi_i x) 
+ \beta_i \gamma_i \sin(2 \pi \xi_i x),
\end{align}
where for all $i,j = 1,...,m$
\begin{align}
\alpha_i, \beta_j \sim \cN (0, 1)
\end{align}
are iid and $\gamma_i$ are defined by
\begin{align}
\gamma_i = \sqrt{2 w_i \hk(\xi_i)}
\end{align}
for some $\xi_i, w_i > 0$.
Then $f$ is a Gaussian process distribution with covariance kernel $k'$ defined by the formula
\begin{align}
k'(x) = \sum_{i=1}^m 2 w_j \hk(\xi_j) \cos(2\pi \xi_j x).
\end{align}
\end{theorem} 
\begin{proof}
Using the independence of the Gaussian coefficients of $f$, clearly, 
\begin{equation}\label{111}
\begin{aligned}
E[f(x) f(y)] & = \sum_{j=1}^m 2w_j \hk(\xi_j) \cos(2\pi\xi_j x)\cos(2\pi\xi_j y)
+ \sum_{j=1}^m 2w_j \hk(\xi_j) \sin(2\pi\xi_j x)\sin(2\pi\xi_j y) \\
& =  \sum_{j=1}^m 2 w_i \hk(\xi_j) 
\bigg( \cos(2\pi\xi_j x)\cos(2\pi\xi_j y) + \sin(2\pi\xi_j x)\sin(2\pi\xi_j y) \bigg)
\end{aligned}
\end{equation}
Applying standard trigonometric properties to (\ref{111}), we obtain
\begin{align}
E[f(x) f(y)] 
& = \sum_{j=1}^m 2 w_j \hk(\xi_j) \cos(2\pi\xi_j (x-y)).
\end{align}
\end{proof}

An immediate consequence of Theorem \ref{120} is that an $m$-point quadrature rule for the evaluation of the inverse Fourier transform of a covariance kernel provides a weight-space representation of that Gaussian process. Moreover the accuracy of the quadrature rule is exactly the accuracy of the effective covariance kernel. Specifically, 
\begin{align}\label{111b}
\bigg| \int_{0}^{\infty} 2 \hk(\xi) \cos(2\pi \xi x) d\xi - \sum_{j=1}^m 2 w_j \hk(\xi_j) \cos(2\pi \xi_j x) \bigg| = | k(x) - k'(x) |.
\end{align}
Due to the computational demands of computing with Gaussian processes,
practitioners often make tradeoffs between accuracy of some approximate
inference algorithm and computational efficiency. 
In the algorithm of this paper, we use quadrature rules, discussed in subsequent sections, that are equipped with theoretical guarantees on their accuracy. These guarantees on quadrature accuracy translate directly to guarantees on pointwise kernel approximation error due to \eqref{111b}. We now describe the numerical procedure we use for constructing quadrature rules for Gaussian process representation \eqref{56}.

\subsection{Gaussian quadratures for covariance kernels}
In the numerical scheme of~\cite{bremer1} for constructing generalized Gaussian 
quadratures, the user inputs functions $\phi_1,...,\phi_{n} : [a, b] \rightarrow \R$ for 
some $n > 1$ and is returned the nodes $x_1,...,x_m \in [a, b]$ and weights 
$w_1,...,w_m \in \R^+$ such that 
\begin{align}\label{121}
\bigg| \int_{a}^{b} \phi_j(x) dx  
- \sum_{i=1}^m w_i \phi_j(x_i) \bigg| < \epsilon
\end{align}
for some user-specified $\epsilon > 0$ for all $j = 1,...,n$. 

In this paper, we are concerned with constructing a particular class of quadrature rules. If $f$ is a Gaussian process defined on $[a, b]$ with covariance kernel $k$, then we seek to approximate integrals of the form 
\begin{align}\label{1723}
\int_{0}^{\infty} 2 \hk(\xi) \cos(2\pi \xi t) d\xi
\end{align}
for all $t \in [0, b-a]$, which contains $\{|x - y| : x, y \in [a, b]\}$. We therefore use the numerical scheme in ~\cite{bremer1} to construct quadrature rules for the set of functions 
\begin{align}
\phi_j(\xi) = 2\hk(\xi) \cos(2\pi \xi t_j)
\end{align}
for $j =1, 2, ..., n$ and $\xi \in [0, \infty)$. Since the integrals in (\ref{1723})
are smooth functions in $t$, it is sufficient to choose $t_1,...,t_n$ as, for example, order-$n$ Chebyshev nodes on $[0, b-a]$~\cite{trefethen}, provided that $n$ is sufficiently large. 
This procedure can be viewed as expanding a smooth function in Chebyshev series where computation is done in $\epsilon$-precision arithmetic. 
Specifically, suppose that we construct a quadrature rule with nodes $\xi_1,...,\xi_m > 0$ and $w_1,...,w_m > 0$ such that
 \begin{align}
\bigg| \int_{0}^{\infty} 2 \hk(\xi) \cos(2\pi \xi t_j) d\xi  
- \sum_{i=1}^m 2 w_i \hk(\xi_i) \cos(2\pi \xi_i t_j) \bigg| < \epsilon
\end{align}
for some user-specified $\epsilon > 0$ and $j=1,...,n$. Then for sufficiently large $n$,
 \begin{align}
\bigg| \int_{0}^{\infty} 2 \hk(\xi) \cos(2\pi \xi t) d\xi  
- \sum_{i=1}^m 2 w_i \hk(\xi_i) \cos(2\pi \xi_i t) \bigg| < \epsilon
\end{align}
for all $t \in [0, b-a]$.

Similarly, we can use the same strategy to discretize a family of covariance functions over ranges of hyperparameters, provided that $\hk$ is a smooth function of those hyperparameters. For example, suppose that $f$ is a GP defined on $[-1, 1]$ with a Mat\'ern kernel $k_{\nu, \rho}$ where $k_{\nu, \rho}$ and its Fourier transform $\hk_{\nu, \rho}$ are defined by 
\begin{equation}\label{mat_ker}
\begin{aligned}
k_{\nu, \rho}(x) & = \frac{2^{1-\nu}}{\Gamma(\nu)}
\bigg(  \sqrt{2\nu}\frac{x}{\rho}\bigg)^{\nu}  K_{\nu}\bigg(\sqrt{2\nu}\frac{x}{\rho}\bigg) \qquad \text{and}\\
\hk_{\nu, \rho}(\xi) & := \int_{\R} \hk(\xi) e^{2\pi i\xi x}d\xi \propto \frac{\Gamma(\nu + \frac{1}{2})}{\Gamma(\nu) \rho^{2\nu}} \bigg( \frac{2\nu}{\rho^2} \bigg)^{\nu} \bigg(\frac{2\nu}{\rho^2} + 4\pi^2 \xi^2 \bigg)^{-(\nu + 1/2)},
\end{aligned}
\end{equation}
where $\nu \geq 1/2$, $\rho$ is the lengthscale, and $K_{\nu}$ is the modified Bessel function of the second kind. 
Now consider the set of covariance kernels $ k_{\nu_i, \rho_j} $ for $i,j = 1,...,p$ where $\nu_1,...,\nu_p$ are the order-$p$ Chebyshev nodes on $[3/2, 7/2]$ and $\rho_1,...,\rho_p$ are the order-$p$ Chebyshev nodes defined on $[0.1, 0.5]$. Then $\hk$ is a smooth function of $\nu, \rho$ over their domain and we can use~\cite{bremer1} to construct quadrature rules for the family of integrals
\begin{align}
\int_{0}^{\infty} 2 \hk_{\nu_i, \rho_j} (\xi) \cos(2\pi \xi t_{\ell}) d\xi
\end{align} 
for all $i, j \in \{1,...,p\}$ and $\ell \in \{1,...,n\}$ up to some tolerance $\epsilon$. 
For large enough $p$ and $n$, the resulting quadrature rules satisfy
\begin{align}
\bigg| \int_{0}^{\infty} 2 \hk_{\nu, \rho}(\xi) \cos(2\pi \xi t) d\xi  
- \sum_{i=1}^m 2 w_i \hk_{\nu, \rho}(\xi_i) \cos(2\pi \xi_i t) \bigg| < \epsilon
\end{align}
for all $t \in [0, 2]$, $\rho \in [0.1, 0.5]$, $\nu \in [3/2, 7/2]$. 

We used the implementation of~\cite{serkh1} for this particular computation with $p = 100$ and $n = 200$. The code took $88$ seconds to run on a laptop and returned $86$ nodes $\xi_1,...,\xi_{86}$ and positive weights $w_1,...,w_{86}$.  
%Additionally, the quadrature was Gaussian -- the numerical rank of the inputted space of functions was twice the number of outputted nodes. 
The results of numerical experiments using this quadrature rule are included in Section \ref{s120}. We also include in Section \ref{s120} results of numerical experiments 
with squared-exponential kernels, $k$, defined by 
\begin{align}\label{se_def}
k(x) = e^{\frac{-x^2}{2\rho^2}},
\end{align}
where $\rho$ is a lengthscale hyperparameter. 

We now describe an algorithm for constructing Fourier representations of Gaussian processes for a range of hyperparameter values. 

\begin{algorithm}\label{a10}[Construction of Fourier representations]

\begin{enumerate}

\item Set the interval on which the Gaussian process is defined, $[a, b]$, in addition to the intervals where the hyperparameters $\nu \in [\nu_0, \nu_1]$ and $\rho \in [\rho_0, \rho_1]$ are defined. Additionally, set the error tolerance $\epsilon$ for the accuracy of the quadrature, or equivalently accuracy of the effective covariance kernel. 

\item Construct a quadrature rule over the region of interest using the algorithm of \cite{bremer1}. Specifically, find nodes $\xi_1,...,\xi_m \in \R$ and weights $w_1,...,w_m \in \R^+$ such that 
\begin{align}\label{113}
\bigg| k(x) - \sum_{i=1}^m 2 w_i \hk(\xi_j) \cos(2\pi \xi_j x) \bigg| < \epsilon
\end{align}
for all $\nu \in [\nu_0, \nu_1]$ , $\rho \in [\rho_0, \rho_1]$, and $x \in [0, b-a]$.

\item Define $f: [a, b] \rightarrow \R$ to be the random expansion
\begin{align}\label{122}
f(x) \sim  \sum_{i=1}^{m} \alpha_i \gamma_i \cos(2 \pi \xi_i x) 
+ \beta_i \gamma_i \sin(2 \pi \xi_i x),
\end{align}
where $\gamma_i$ are defined in \eqref{gamma_i} and $\alpha_i$ and $\beta_i$ are iid
standard normal Gaussians. 
Then $f$ is a Gaussian process with effective covariance kernel $k'$ defined by 
\begin{align}
k'(x) = \sum_{i=1}^n 2w_i \hk(\xi_j) \cos(2\pi \xi_j x).
\end{align} 
\end{enumerate}

\end{algorithm}

\section{Regression}\label{s20}
Typically, GP regression is used in the following environment. An applied scientist has observed data $\{(x_i, y_i)_{i=1,...,N}\}$ where $x_i$ are independent variables that belong to some interval $[a, b]\in \R$ (for $1$-dimensional problems) and $y_i \in \R$ are dependent variables. The observed data $y_1,...,y_N$ are assumed to be observations of the form 
\begin{align}
y_i = f(x_i) + \epsilon_i,
\end{align}
where $\epsilon_i$ is iid Gaussian noise and $f: [a, b] \to \R$ is an unknown function, which is given a Gaussian process prior with covariance function $k$. Assumptions about $k$ are critical for statistical inference and typically arise from domain expertise or physical knowledge about the data-generating process. There is a large body of literature on the selection of suitable covariance kernels (e.g.~\cite{vehtari1, stein2, duvenaud1}).  In many applied settings $k$ is not known a priori, but is assumed to belong to some parametric family of functions (e.g. Mat\'ern, squared-exponential, etc.) that depends on hyperparameters that are fit to the data. 

The goal of GP regression is to perform statistical inference on the unknown function $f$ at a set of points $\tilde{x} \in [a, b]$ or to understand certain properties of the data-generating process. Inference typically involves evaluating the mean and covariance of the density, conditional on observed data $\vct{x}, \vct{y} \in \R^{N}$. The conditional (or posterior) distribution of $f$ at any $\tilde{x} \in \R$ is the Gaussian
\begin{align}
f(\tilde{x}) \, | \, \vct{x}, \vct{y} \sim \cN(\tilde{\mu}, \tilde{\sigma}^2)
\end{align}
with 
\begin{equation}\label{84}
  \begin{aligned}
& \tilde{\mu} = \vct{k}(\tilde{x}, \vct{x}) (\mtx{K} + \sigma^2 \mtx{I})^{-1} \vct{y},  \\
& \tilde{\sigma}^2 = k(\tilde{x}, \tilde{x}) - \vct{k}(\tilde{x}, \vct{x}) (\mtx{K} + \sigma^2 \mtx{I})^{-1} \vct{k}(\vct{x}, \tilde{x}),    
\end{aligned}
\end{equation} 
where $\sigma$ is the standard deviation of $\epsilon_i$ (also called the nugget), the matrix $\mtx{K}$ is the $N \times N$ covariance matrix $\mtx{K}_{i,j} = k(x_i, x_j)$, the vector $\vct{k}(\vct{x}, x') \in \R^N$ is the column vector such that $\vct{k}(\tilde{x}, \vct{x})_i = f(\tilde{x}, \vct{x}_i)$, and $\vct{k}(\vct{x}, \tilde{x}) = \vct{k}(\tilde{x}, \vct{x})\tran$. 

Since $\mtx{K} + \sigma^2\mtx{I}$ is an $N \times N$ matrix, direct inversion is computationally intractable for large $N$. However, the Fourier representations of this paper admit a natural low-rank approximation to $\mtx{K}$. Specifically, $\mtx{K}$ is well-approximated by $\mtx{XX\tran}$ where $\mtx{X}$ is the $N \times 2m$ matrix defined in \eqref{x_mat}. In particular, $\mtx{K}_{i, j} = k(x_i, x_j)$ and $\mtx{XX\tran}_{i, j}$ is the quadrature rule approximation to $k(x_i, x_j)$ given by
\begin{align}
\sum_{\ell=1}^m 2\hk(\xi_{\ell}) w_{\ell} \cos(2\pi \xi_{\ell} (x_i - x_j)).
\end{align}
Using $\mtx{X X\tran}$ as an approximation of $\mtx{K}$, the linear systems that appear in $\tilde{\mu}$ and $\tilde{\sigma}^2$,
\begin{align}\label{138}
(\mtx{K} + \sigma^2 \mtx{I}) \vct{x} = \vct{y} \quad \text{and} \quad (\mtx{K} + \sigma^2 \mtx{I}) \vct{x} = \vct{k}(\vct{x}, \tilde{x}),
\end{align}
can be approximated in $O(Nm^2)$ operations using standard direct methods. 
For problems where $Nm^2$ operations is computationally infeasible, 
linear systems \eqref{138} can be obtained by solving the weight-space 
systems described in the following section \cite{greengard3}. 

\subsection{Weight-space inference}
In the weight-space approach to inference, the posterior density is defined over the coefficients of the basis function expansion -- in this paper a Fourier expansion. Specifically, given data $\{(x_i, y_i)\}$ and a covariance kernel $k$, Gaussian process regression becomes 
\begin{equation}\label{721}
\begin{aligned}
\vct{y} & \sim \mtx{X}\vct{\beta} + \vct{\epsilon} \\ % \sum_{j=1}^k (\alpha_j\cos(\xi_j x_i) + \beta_j \sin(\xi_j x_i))  \\
\vct{\epsilon} & \sim \cN(0, \sigma^2 \mtx{I}) \\
\vct{\beta} & \sim \cN(0, \mtx{I}),
\end{aligned}
\end{equation}
where $\vct{\beta} \in \R^{2m}$ is the vector of coefficients and $\mtx{X}$ is the $N \times 2m$ matrix 
(\ref{x_mat}). The posterior density corresponding to (\ref{721}) is defined over $2m$ dimensions, the components of $\vct{\beta}$ (or $2m + 1$ including the residual standard deviation $\sigma$). The conditional
distribution of $\vct{\beta}$ (conditioning on $\sigma, \vct{x}, \vct{y}$) is Gaussian. 
Specifically,
\begin{align}
\vct{\beta} \, | \, \sigma, \vct{x}, \vct{y} \sim \cN(\vct{\bar{\beta}}, (\mtx{X\tran X} + \sigma^2 \mtx{I})^{-1}),
\end{align}
where $\vct{\bar{\beta}}$ is the solution to the linear system of equations
\begin{align}
(\mtx{X\tran X} + \sigma^2 \mtx{I}) \vct{\bar{\beta}} = \mtx{X\tran}\vct{y}.
\end{align}
The conditional mean $\vct{\bar{\beta}} \in \R^{2m}$ is the vector of coefficients of the 
conditional mean function
\begin{align}\label{cond_mean}
\sum_{i=1}^{m} \vct{\bar{\beta}}_{1,i} \cos(2\pi \xi_i x) + \vct{\bar{\beta}}_{2,i} \sin(2\pi \xi_i x)
\end{align}
for all $x \in [a, b]$. 

We now describe a numerical implementation of a solver for this system of equations that 
requires $O(N + m^3)$ operations by taking advantage of the Fourier representation of GPs.

\subsection{Numerical Implementation}
For a general $N \times 2m$ matrix $\mtx{X}$, the computational complexity of the numerical solution of the linear system of equations 
\begin{align}\label{91}
(\mtx{X\tran X} + \sigma^2 \mtx{I} ) \vct{x} = \mtx{X\tran}\vct{y}
\end{align}
is $O(Nm^2)$. For GP regression tasks with large amounts of data, $O(Nm^2)$ can be prohibitively expensive. However, for the linear system that appears in the case of Fourier representations, we take advantage of the structure of $\mtx{X}$ to reduce the linear solve to $O(N + m^3)$ operations. We do this by constructing the $2m \times 2m$ matrix $\mtx{X\tran X}$ in $O(N + m^2\log{m})$ operations using a non-uniform fast Fourier transform (FFT)~\cite{dutt1, greengard2} whereas constructing $\mtx{X\tran X}$ ordinarily requires $O(Nm^2)$ operations. After constructing $\mtx{X\tran X}$, the cost of the linear solve, for all hyperparameters input into Algorithm \ref{a10}, is $O(m^3)$. 

We observe that $\mtx{X\tran X}$ can be factorized as
\begin{align}\label{97}
\mtx{X\tran X} = \mtx{B\tran D \mtx{X'}\tran \mtx{X'} D B},
\end{align}
where $\mtx{X'}$ is defined by
\begin{equation}
\mtx{X'} = 
\renewcommand*{\arraystretch}{1.2}
\begin{bmatrix}
    e^{2\pi i\xi_1 x_1} & \dots & e^{2\pi i\xi_m x_1} & e^{2\pi i (-\xi_1) x_1} & \dots & e^{2\pi i (-\xi_m) x_1} \\
    e^{2\pi i\xi_1 x_2} & \dots & e^{2\pi i\xi_m x_2} & e^{2\pi i (-\xi_1) x_2} & \dots & e^{2\pi i (-\xi_m) x_2} \\
    \vdots &  & \vdots & \vdots & & \vdots \\
    e^{2\pi i\xi_1 x_N} & \dots & e^{2\pi i\xi_m x_N} & e^{2\pi i (-\xi_1) x_N} & \dots & e^{2\pi i (-\xi_m) x_N} \\
\end{bmatrix},
\end{equation}
where 
$\mtx{D}$ is the diagonal $2m \times 2m$ matrix
\begin{equation}\label{117}
\mtx{D} = 
\renewcommand*{\arraystretch}{1.2}
\begin{bmatrix}
    \gamma_1 & &&&&&\\
    & \ddots &&&&&\\
    & & \gamma_m&&&&\\
    & & & \gamma_1 &&&\\
    & & & & \ddots & &\\
    & & & & & \gamma_m &\\
\end{bmatrix},
\end{equation}
(see \eqref{gamma_i} for the definition of $\gamma_i$) and $\mtx{B}$ is $2m \times 2m$ block matrix 
\begin{equation}\label{118}
\mtx{B} = 
 \renewcommand*{\arraystretch}{1.5}
 \begin{bmatrix}
    \frac{1}{2} \mtx{I}_{m} & \frac{1}{2i}\mtx{I}_{m} \\ 
    \frac{1}{2} \mtx{I}_{m} & -\frac{1}{2i}\mtx{I}_{m}
\end{bmatrix},
\end{equation}
where $\mtx{I}_{m}$ denotes the $m \times m$ identity matrix. 
The matrix $\mtx{X'}\tran\mtx{X'}$ is a symmetric matrix where entry ${p, q}$ is given by 
\begin{align}\label{223}
(\mtx{X'}\tran \mtx{X'})_{p, q} = \sum_{j=1}^{N} e^{2\pi i x_j (\xi_{p} + \xi_{q})},
\end{align}
for all $p, q \in \{1,...,2m\}$, where we denote $-\xi_{p}$ with $\xi_{m + p}$. The sums of (\ref{223}) can be evaluated with a type 3 non-uniform FFT, a calculation that requires $O(N + m^2\log{m})$ operations. 
More precisely, we use the type 3 non-uniform FFT of~\cite{barnett1} to compute the sums
\begin{align}\label{eq50}
\sum_{j=1}^N e^{2\pi  i x_j \omega_{\ell}},
\end{align}
where $\ell = 1,...,2m^2+m$ and $\omega_{\ell} = \xi_p + \xi_q$ for $p, q \in \{1,...,2m\}$ with $q \geq p$. The total cost of evaluating \eqref{eq50} via the non-uniform FFT is $O(N + m^2\log{m})$ operations. A detailed description of the algorithm and its computational costs can be found in \cite{barnett1}. 
We also use the non-uniform FFT to compute $\mtx{X\tran}\vct{y}$, the right hand side of (\ref{91}). 

Since matrix multiplications of (\ref{97}) other than $\mtx{X'}\tran \mtx{X'}$ can be applied in $O(m^2)$ operations, the total cost of forming matrix $\mtx{X\tran X}$ using a non-uniform FFT is $O(N + m^2\log{m})$. Once $\mtx{X\tran X}$ is formed, the solution to the linear system and determinant calculation can be obtained in $O(m^3)$ operations using standard direct methods. In the following, we provide an algorithm for solving linear system (\ref{91}) in $O(N + m^3)$ operations.

\begin{algorithm}\label{a20}[GP regression solver]
\begin{enumerate}
\item
Use Algorithm \ref{a10} to construct nodes $\xi_1,...,\xi_m$ and weights $w_1,...,w_m$ of a Fourier expansion for a certain covariance kernel or family of kernels. 

\item\label{213}
Use the non-uniform FFT to compute the matrix-vector product
\begin{align}
\mtx{X\tran}\vct{y}
\end{align}
that appears in the right hand side of (\ref{91}). 

\item\label{215}
Use the non-uniform FFT to compute the sums 
\begin{align}
\sum_{j=1}^N e^{2\pi ix_j \omega_{\ell}}
\end{align}
where $\ell = 1,...,2m^2+m$ and $\omega_{\ell} = \xi_p + \xi_q$ for $p, q \in \{1,...,2m\}$ with $q \geq p$. 

\item 
Construct $\mtx{X\tran X}$ via 
\begin{align}
\mtx{X\tran X} = \mtx{B}\tran \mtx{D} \mtx{X'}\tran \mtx{X'} \mtx{D} \mtx{B},
\end{align}
where $\mtx{D}$ is defined in (\ref{117}) and $\mtx{B}$ is defined in (\ref{118}).

\item\label{217}
Compute the eigendecomposition (or Cholesky factorization) of $\mtx{X\tran X}$. 

\item
For computing the posterior mean or variance of GP regression, solve the linear system 
\begin{align}\label{ws_sys5}
(\mtx{X\tran X} + \sigma^2 \mtx{I}) \vct{\beta} = \mtx{X\tran}\vct{y}. 
\end{align}
The determinant of the matrix of \eqref{ws_sys5} can be evaluated via the 
matrix decomposition computed in step \ref{217}. 

\end{enumerate}
\end{algorithm} 

In Table \ref{7802} we provide computation times for solving the linear system of equations (\ref{91}) as well as forming $\mtx{X\tran X}$ and the matrix vector multiply $\mtx{X\tran}\vct{y}$. Additionally, in Figure \ref{2200} we plot the total computation times for GP regression for varying numbers of observations.

\subsection{Adaptation of hyperparameters}
In certain GP problems, the covariance function, $k$ is known a priori. However, in general, $k$ is assumed to belong to a certain parametric family of functions (e.g. Mat\'ern, squared-exponential, etc.) that depends on hyperparameters that are fit to the data. Methods for fitting hyperparameters to data generally involve performing GP regression and solving the linear system of equations (\ref{91}) and calculating a determinant for many hyperparameter 
values. 

For those problems, the methods of this paper have several advantages:
\begin{itemize}
\item
Algorithm \ref{a10} can be used to generate a Fourier expansion that is valid over the  domain of hyperparameters. That expansion depends on hyperparameters by a rescaling of basis functions (or equivalently by rescaling the prior standard deviation of the coefficients).

\item
Using Algorithm \ref{a20}, GP regression and determinant evaluation are performed in $O(m^3)$ operations for all hyperparameters after a precomputation of $O(N + m^2\log{m})$ operations. The precomputation involves the application of non-uniform FFTs (steps \ref{213} and \ref{215} of Algorithm \ref{a20}). 

\item
Since both the observation noise and the priors on regression coefficients are Gaussian, the posterior density in the coefficients $\beta$ is also Gaussian. As a result efficient numerical methods can be used for evaluating Bayesian posterior moments~\cite{lindley1, greeng1, greengard2} for certain problems.

\item
Evaluation of gradients of the likelihood function can be performed in $O(m^2)$ operations after inversion of the posterior covariance matrix (step \ref{217} in Algorithm \ref{a20}). The component-wise formula for the gradient of the log-likelihood of the posterior density is given by
\begin{align}\label{219}
\frac{\partial}{\partial \theta_j} \log(p(\vct{y} | \vct{\theta})) = \frac{1}{2}\vct{y\tran} \mtx{C^{-1}} \frac{\partial \mtx{C}}{\partial \theta_j}
\mtx{C^{-1}} \vct{y} - \frac{1}{2}\text{tr}\bigg(\mtx{C^{-1}}\frac{\partial \mtx{C}}{\partial \theta_j}\bigg),
\end{align}
where $\theta_j \in \R$ is the $j$-th hyperparameter (see e.g. \cite{rasmus1}), and $\mtx{C} = \mtx{X\tran}\mtx{X} + \sigma^2 \mtx{I} = \mtx{B}\tran \mtx{D} \mtx{X'}\tran \mtx{X'} \mtx{D} \mtx{B} + \sigma^2 \mtx{I}$ (see \eqref{97}). That is, the gradient of the log-likelihood can be evaluated in a time independent of $N$, the number of data points, after the one-time $O(N + m^2\log{m})$ cost of construction of $\mtx{X'}\tran \mtx{X'}$. 
\end{itemize}

\section{Numerical Experiments}\label{s120}
In this section we demonstrate the performance of the algorithms of this paper on randomly generated data. We implemented Algorithm \ref{a10} and Algorithm \ref{a20} in Fortran with the GFortran compiler on a 2.6 GHz 6-Core Intel Core i7 MacBook Pro. All examples were run in double precision arithmetic.

We used Algorithm \ref{a10} to construct Fourier expansions for the Gaussian processes defined on $[-1, 1]$ with Mat\'ern kernel (see \eqref{mat_ker})
for all $\nu \in [3/2, 7/2]$ and $\rho \in [0.1, 0.5]$.
We also generated quadratures for the squared-exponential kernel (see \eqref{se_def})
with $\rho \in [0.1, 0.5]$ and two quadrature tolerances, $\epsilon = 10^{-3}, 10^{-5}$ (see Step 2 of Algorithm \ref{a10}). 
We chose these kernels due to their widespread use in applications\cite{rasmus1}. 
The Mat\'ern family's popularity is largely due to its flexibility in representing processes 
of varying smoothness. In particular, the parameter $\nu$ controls smoothness -- a
Gaussian process with Mat\'ern kernel is $\lfloor \nu \rfloor$ times mean-squared
differentiable. Our choice of hyperparameters for the Mat\'ern kernel represents 
processes with a range of smoothness properties. 
We used the implementation of \cite{serkh1} for the generalized Gaussian quadrature with 
Mat\'ern and squared-exponential kernel. For the Mat\'ern kernel, we set $\epsilon = 10^{-5} $ in Algorithm \ref{a10}. The total run time for generating the  quadrature was $88$ seconds. 
The output of the code was $86$ total nodes and weights and all weights were positive. 
For the squared-exponential kernel, we generated two generalized Gaussian quadratures, 
one with $\epsilon = 10^{-5}$ and one with $\epsilon = 10^{-3}$. 
These procedure each took $13$ seconds and generated $21$ and $16$ nodes
respectively.

These numerical experiments demonstrate that the number of nodes, $m$, required to 
discretize the squared-exponential family of kernels is significantly smaller than 
the $m$ needed for the Mat\'ern family. 
Although both families need a similar density of nodes near zero to integrate kernels 
with large timescales $\rho$ (which concentrate near zero in Fourier domain), the 
Mat\'ern family requires more nodes at high frequencies. This is because
Mat\'ern kernels have slower decay in Fourier domain, with decays ranging from 
fourth to eigth order for $\rho \in [3/2, 7/2]$, whereas squared-exponential kernels 
have Gaussian decay. 
Therefore, while the squared-exponential and Matérn quadratures need similar 
densities of points near zero for accurate integration, the Matérn quadrature 
needs many more points at large frequencies

Figure \ref{2270} is a plot of the locations of the nodes obtained from this procedure. In Table \ref{7813} we list those nodes and the corresponding weights. 
Figure \ref{2271} is the corresponding plot for the squared-exponential quadratures 
and Tables \ref{7814} and \ref{t:se_nodes_weights_1e3} are the tables of nodes and weights.

In Figure \ref{fig:matern_heatmap} and Table \ref{7801bc}, we provide the $L^2$ error of 
the effective covariance
kernel for various hyperparameter values. The $L^2$ error of the effective kernel is defined
by the formula
\begin{align}\label{133}
\bigg(
\int_{-1}^{1}
\int_{-1}^{1}
(k'(x, y) - k(x, y))^2
dx \, dy
\bigg)^{1/2},
\end{align}
where $k'$ denotes the effective covariance kernel and $k$ denotes the true kernel. Integral (\ref{133}) was computed numerically using a tensor product of Gaussian nodes and weights. 
The approximate GP regression algorithm of this paper can be viewed as exact GP regression with a kernel that approximates the user-specified kernel. Here, we choose $L^2$ to measure the accuracy of that approximation. We note that it was shown in \cite{barnett2024} that a uniform bound on the pointwise error $|k(x) - k'(x)| < \epsilon$ for all $x$ implies a bound on the $l^2$ error in the conditional mean \eqref{cond_mean} at the data points of $N \epsilon / \sigma^2$.

Tables \ref{7802} and \ref{7802c} contain the results of numerical experiments for Gaussian process regression using Algorithms \ref{a10} and \ref{a20}. The data, $\{(x_i, y_i)\}$ was randomly generated on the interval $[-1, 1]$ via 
\begin{align}
y_i = \cos(3e^{x_i}) + \epsilon_i,
\end{align}
where 
\begin{align}
\epsilon_i \sim \cN(0, 0.5).
\end{align}
The $x_i$ were generated uniformly at random on $[-1, 1]$ and we note that the performance of the algorithm
is independent of locations of data points. The nodes and weights of the Fourier expansions used are given by Table \ref{7813}. We used Algorithm \ref{a20} to compute the conditional mean and covariance and report timings in Tables \ref{7802} and \ref{7802c} where $N$ denotes the number of data points, $\nu$ denotes the smoothness hyperparameters of the Mat\'ern kernel, and $\rho$ denotes the lengthscale. We used a residual variance, $\sigma^2 = 1$ and note that the performance of the algorithm is independent of this parameter. The $L^2$ error of the effective covariance kernel (see (\ref{133})) is provided in the column ``$L^2$ error". We include timings for solving the regression problem via Algorithm \ref{a20}. In column ``$t_{\text{form}}$ (s)" we provide the total time for formation of the matrix $\mtx{X\tran X}$ and the matrix vector multiply $\mtx{X\tran} \vct{y}$. The column ``$t_{\text{solve}}$ (s)" denotes the remaining time in the regression solution, which involves inversion of $(\mtx{X\tran X} + \sigma^2 \mtx{I})$. The total time is provided in column ``$t_{\text{total}}$ (s)."
The primary purpose of Tables \ref{7802} and \ref{7802c} is to demonstrate the 
scaling of compute times for $N$ varying from $10^5$ to $10^8$. 
These timings 
are independent of the choice of hyperparameters and thus hyperparameters for 
these experiments were arbitrarily chosen. 

The numerical experiments in this section illustrate the two primary advantages of 
the algorithms of this paper. 
First, in Table 1 and in Figure 3, we show that generalized quadratures
are able to approximate kernels (Mat\'ern and squared-exponential) over ranges of hyperparameters to high accuracy. 
Second, in Tables \ref{7802}, \ref{7802c} and Figure \ref{2200}, we demonstrate
that the solution to the GP linear system can be efficiently computed with fast 
algorithms. Since $m$ is significantly smaller than $N$, nearly all compute time 
in the linear solve is in the $O(N + m^2\log{m})$ precomputation consisting of
two non-uniform FFTs for forming $\mtx{X\tran}\mtx{X}$ and $\mtx{X\tran}\vct{y}$. 
Furthermore, after precomputation the cost of all linear solves is $O(m^3)$. 
We report timings in column $t_{\text{solve}}$ in Tables \ref{7802} and 
\ref{7802c}.

\begin{figure}[!ht]
\centering
\begin{tikzpicture}
\begin{axis}[
width=17cm,
height=3cm,
    xlabel={$\xi$},
    xmin = -1,
    xmax=51,
    ymin=-0.5,
    ymax=0.5,
    ytick=\empty,
]
\addplot[only marks, mark=*]
    coordinates 
    {
(    0.0960748783232733,    0.0000000000000000)
(    0.2949311558758212,    0.0000000000000000)
(    0.5121811831688609,    0.0000000000000000)
(    0.7553750530381257,    0.0000000000000000)
(    1.0272554775524241,    0.0000000000000000)
(    1.3267469591800769,    0.0000000000000000)
(    1.6509109783729810,    0.0000000000000000)
(    1.9964656597608781,    0.0000000000000000)
(    2.3604331716760738,    0.0000000000000000)
(    2.7402749478447399,    0.0000000000000000)
(    3.1338448880682450,    0.0000000000000000)
(    3.5393213277133500,    0.0000000000000000)
(    3.9551382871240710,    0.0000000000000000)
(    4.3799500415428971,    0.0000000000000000)
(    4.8125843418627330,    0.0000000000000000)
(    5.2520252807076444,    0.0000000000000000)
(    5.6973962780821781,    0.0000000000000000)
(    6.1479398698978196,    0.0000000000000000)
(    6.6029703189280022,    0.0000000000000000)
(    7.0619657471053916,    0.0000000000000000)
(    7.5244050646589304,    0.0000000000000000)
(    7.9898781249339503,    0.0000000000000000)
(    8.4580188390522402,    0.0000000000000000)
(    8.9284697412154603,    0.0000000000000000)
(    9.4010606000280621,    0.0000000000000000)
(    9.8755021184860947,    0.0000000000000000)
(   10.3515866576509996,    0.0000000000000000)
(   10.8290813405759305,    0.0000000000000000)
(   11.3078856905533307,    0.0000000000000000)
(   11.7878547186769005,    0.0000000000000000)
(   12.2690444362196605,    0.0000000000000000)
(   12.7507923286660407,    0.0000000000000000)
(   13.2340308273640499,    0.0000000000000000)
(   13.7168180215943405,    0.0000000000000000)
(   14.2023546685620694,    0.0000000000000000)
(   14.6852339737999902,    0.0000000000000000)
(   15.1736671929864606,    0.0000000000000000)
(   15.6561205601530808,    0.0000000000000000)
(   16.1465532564475396,    0.0000000000000000)
(   16.6287678133340897,    0.0000000000000000)
(   17.1217168205720291,    0.0000000000000000)
(   17.6027784484422583,    0.0000000000000000)
(   18.0975219052295202,    0.0000000000000000)
(   18.5783668916359304,    0.0000000000000000)
(   19.0739289334232396,    0.0000000000000000)
(   19.5535187961381212,    0.0000000000000000)
(   20.0540157049340202,    0.0000000000000000)
(   20.5303087220011804,    0.0000000000000000)
(   21.0308285956013101,    0.0000000000000000)
(   21.5033438976054896,    0.0000000000000000)
(   22.0147025165826093,    0.0000000000000000)
(   22.4786227659423616,    0.0000000000000000)
(   22.9890428278831500,    0.0000000000000000)
(   23.4620311240316504,    0.0000000000000000)
(   23.9659051179655798,    0.0000000000000000)
(   24.4374955606510014,    0.0000000000000000)
(   24.9268582817982995,    0.0000000000000000)
(   25.4355862458376691,    0.0000000000000000)
(   25.8908822711706499,    0.0000000000000000)
(   26.4135765170375585,    0.0000000000000000)
(   26.8244665781450813,    0.0000000000000000)
(   27.3929596371526216,    0.0000000000000000)
(   27.8810764015625807,    0.0000000000000000)
(   28.2872085273741511,    0.0000000000000000)
(   28.9186103516193818,    0.0000000000000000)
(   29.3361902609053793,    0.0000000000000000)
(   29.7475402852299098,    0.0000000000000000)
(   30.2919596593750207,    0.0000000000000000)
(   31.1707837877708087,    0.0000000000000000)
(   31.7921709344750916,    0.0000000000000000)
(   32.0844373549515112,    0.0000000000000000)
(   32.8290834819148998,    0.0000000000000000)
(   33.6870452221527970,    0.0000000000000000)
(   34.6556177080743311,    0.0000000000000000)
(   35.5973697159031417,    0.0000000000000000)
(   36.0608079398967192,    0.0000000000000000)
(   37.4828497993489194,    0.0000000000000000)
(   38.1056387575873927,    0.0000000000000000)
(   38.4560335475650206,    0.0000000000000000)
(   39.9230354160448471,    0.0000000000000000)
(   41.8408663755605872,    0.0000000000000000)
(   43.7213283268615385,    0.0000000000000000)
(   44.2767860363890975,    0.0000000000000000)
(   45.6320859437675992,    0.0000000000000000)
(   47.5466594857420191,    0.0000000000000000)
(   49.4591701554423580,    0.0000000000000000)
    };

\end{axis}
\end{tikzpicture}

\caption{\em Location of the $86$ nodes for GPs defined on $[-1, 1]$ with Mat\'ern kernels with $\nu \in [1.5, 3.5]$, and $\rho \in [0.1, 0.5]$.}
\label{2270}
\end{figure}

\begin{figure}[!ht]
\centering
\begin{subfigure}{\linewidth}
\centering
\begin{tikzpicture}
\begin{axis}[
width=17cm,
height=3cm,
    xlabel={$\xi$},
    xmin = -0.1,
    xmax=7.6,
    ymin=-0.5,
    ymax=0.5,
    ytick=\empty
]
\addplot[only marks, mark=*]
    coordinates 
    {
(    0.1229445208158333,    0.0000000000000000)
(    0.3699809596646247,    0.0000000000000000)
(    0.6205086829190636,    0.0000000000000000)
(    0.8770552409688902,    0.0000000000000000)
(    1.1425176151727869,    0.0000000000000000)
(    1.4200169312069510,    0.0000000000000000)
(    1.7121859512103039,    0.0000000000000000)
(    2.0204572753950010,    0.0000000000000000)
(    2.3450325050947400,    0.0000000000000000)
(    2.6853624988811311,    0.0000000000000000)
(    3.0404027813093690,    0.0000000000000000)
(    3.4089356692419108,    0.0000000000000000)
(    3.7897759561613382,    0.0000000000000000)
(    4.1819675933484106,    0.0000000000000000)
(    4.5848345162589501,    0.0000000000000000)
(    4.9982272332198843,    0.0000000000000000)
(    5.4216263344213846,    0.0000000000000000)
(    5.8568375371290067,    0.0000000000000000)
(    6.3049123731014509,    0.0000000000000000)
(    6.7704615862559896,    0.0000000000000000)
(    7.2856984407800462,    0.0000000000000000)
};
\end{axis}
\end{tikzpicture}
\caption{\em $\epsilon = 10^{-5}$}
\end{subfigure}

\vspace{0.5cm}

\begin{subfigure}{\linewidth}
\centering
\begin{tikzpicture}
\begin{axis}[
width=17cm,
height=3cm,
    xlabel={$\xi$},
    xmin = -0.1,
    xmax=7.6,
    ymin=-0.5,
    ymax=0.5,
    ytick=\empty,
]
\addplot[only marks, mark=*]
    coordinates 
    {
(    0.1422798837388160,    0.0000000000000000)
(    0.4303383428161451,    0.0000000000000000)
(    0.7298461801991711,    0.0000000000000000)
(    1.0470392453484560,    0.0000000000000000)
(    1.3851554070607901,    0.0000000000000000)
(    1.7436383837560350,    0.0000000000000000)
(    2.1204551041309121,    0.0000000000000000)
(    2.5128009567183280,    0.0000000000000000)
(    2.9171012161369898,    0.0000000000000000)
(    3.3326272282803751,    0.0000000000000000)
(    3.7585126968767120,    0.0000000000000000)
(    4.1943700876820937,    0.0000000000000000)
(    4.6430513969235117,    0.0000000000000000)
(    5.0962967069840737,    0.0000000000000000)
(    5.5883652398864117,    0.0000000000000000)
(    6.0596680913512913,    0.0000000000000000)
};
\end{axis}
\end{tikzpicture}
\caption{$\epsilon = 10^{-3}$}
\end{subfigure}

\caption{\em Location of the nodes for GPs defined on $[-1, 1]$ with squared-exponential kernels with $\rho \in [0.1, 0.5]$. Both sets of nodes were generated with Algorithm \ref{a10} with different error tolerance $\epsilon$. }
\label{2271}
\end{figure}

\begin{figure}[b!]
\centering
  \includegraphics[height=10cm, keepaspectratio]{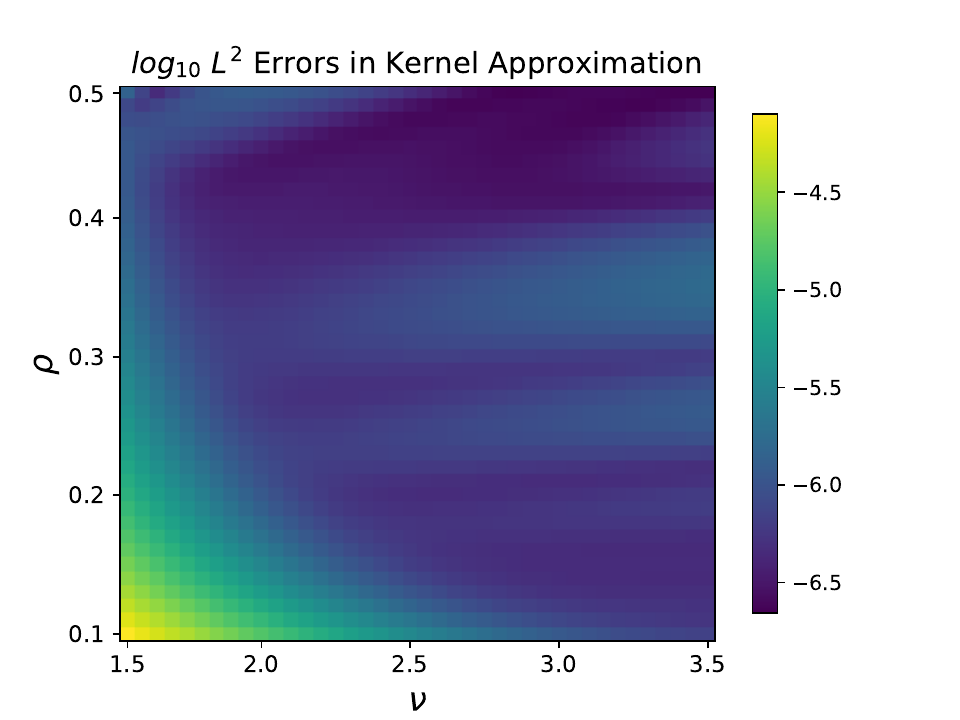}
  \caption{\em $\log_{10} \, \, L^2$ errors in the effective covariance kernel for Mat\'ern kernels. Specifically, we plot $\log_{10} \| k_{\nu, \rho} - k_{\nu, \rho}'\|_2$ for various $\nu, \rho$, where $k'_{\nu, \rho}$ denotes the effective kernel and $k_{\nu, \rho}$ the exact kernel.}
    \label{fig:matern_heatmap}
\end{figure}

\begin{table}[!ht]
\centering
  \begin{subtable}{0.45 \textwidth}
  \centering
\begin{tabular}{cc}
$\rho$ & $L^2$ error \\
  \hline
$0.10$&$0.943\times 10^{-5}$\\
$0.12$&$0.832\times 10^{-5}$\\
$0.14$&$0.847\times 10^{-5}$\\
$0.16$&$0.870\times 10^{-5}$\\
$0.18$&$0.872\times 10^{-5}$\\
$0.21$&$0.855\times 10^{-5}$\\
$0.23$&$0.827\times 10^{-5}$\\
$0.25$&$0.788\times 10^{-5}$\\
$0.27$&$0.732\times 10^{-5}$\\
$0.29$&$0.664\times 10^{-5}$\\
$0.31$&$0.593\times 10^{-5}$\\
$0.33$&$0.537\times 10^{-5}$\\
$0.35$&$0.495\times 10^{-5}$\\
$0.37$&$0.458\times 10^{-5}$\\
$0.39$&$0.421\times 10^{-5}$\\
$0.42$&$0.388\times 10^{-5}$\\
$0.44$&$0.361\times 10^{-5}$\\
$0.46$&$0.339\times 10^{-5}$\\
$0.48$&$0.323\times 10^{-5}$\\
$0.50$&$0.306\times 10^{-5}$\\
\end{tabular}
%}
  \caption{\em $\epsilon=10^{-5}$}
    \label{7801b}
\end{subtable}
\begin{subtable}{0.45 \textwidth}
  \centering
%%%\resizebox{\columnwidth}{!}{%
\begin{tabular}{cc}
$\rho$ & $L^2$ error \\
  \hline
$0.10$&$0.657\times 10^{-3}$\\
$0.12$&$0.646\times 10^{-3}$\\
$0.14$&$0.690\times 10^{-3}$\\
$0.16$&$0.738\times 10^{-3}$\\
$0.18$&$0.780\times 10^{-3}$\\
$0.21$&$0.810\times 10^{-3}$\\
$0.23$&$0.839\times 10^{-3}$\\
$0.25$&$0.856\times 10^{-3}$\\
$0.27$&$0.852\times 10^{-3}$\\
$0.29$&$0.834\times 10^{-3}$\\
$0.31$&$0.823\times 10^{-3}$\\
$0.33$&$0.833\times 10^{-3}$\\
$0.35$&$0.855\times 10^{-3}$\\
$0.37$&$0.872\times 10^{-3}$\\
$0.39$&$0.878\times 10^{-3}$\\
$0.42$&$0.876\times 10^{-3}$\\
$0.44$&$0.872\times 10^{-3}$\\
$0.46$&$0.861\times 10^{-3}$\\
$0.48$&$0.834\times 10^{-3}$\\
$0.50$&$0.805\times 10^{-3}$
\end{tabular}
  \caption{\em $\epsilon=10^{-3}$}
    \label{7801c}
\end{subtable}
  \caption{\em Accuracy of Fourier representation \eqref{23} for squared-exponential kernels with lengthscales $\rho \in [0.1, 0.5]$ and two error tolerances, $\epsilon$. The nodes and weights used are reported in Table \ref{7814} and Table \ref{t:se_nodes_weights_1e3}.}
    \label{7801bc}
\end{table}

\begin{table}[!ht]
  \centering
%%%\resizebox{\columnwidth}{!}{%
\begin{tabular}{lcccccc}
  $N$ & $\nu$ & $\rho$ & $L^2$ error & $t_{\text{form}}$ (s) & $t_{\text{solve}}$ (s) & $t_{\text{total}}$ (s) \\
  \hline
$    10^5 $ & $ 3.0 $ & $ 0.1 $  & $0.113 \times 10^{-5}$ & 0.03 & 0.004 & 0.035 \\
$    10^6 $ & $ 2.0 $ & $ 0.5 $  & $0.118 \times 10^{-4}$ & 0.11 & 0.005 & 0.12 \\
$    10^7 $ & $ 1.5 $ & $ 0.1 $  & $0.780 \times 10^{-4}$ & 1.1 & 0.004 & 1.1 \\
$    10^8 $ & $ 3.5 $ & $ 0.3 $ & $0.630 \times 10^{-6}$ & 12.0 & 0.005 & 12.0
\end{tabular}
%}
  \caption{\em Computation times and accuracy for regression with Mat\'ern kernels with $\nu \in [0.5, 3.5]$, and $\rho \in [0.1, 0.5]$. The nodes and weights used are reported in Table \ref{7813}.}
    \label{7802}
\end{table}

\begin{table}[!ht]
\centering
\begin{subtable}{\linewidth}
  \centering
%%%\resizebox{\columnwidth}{!}{%
\begin{tabular}{lccccccc}
  $N$ & $\rho$ & $L^2$ error & $t_{\text{form}}$ (s) & $t_{\text{solve}}$ (s) & $t_{\text{total}}$ (s) \\
  \hline
$ 10^5$ & $0.50$ & $0.306\times 10^{-5}$ & $0.02$ & $0.1 \times 10^{-3}$ & $0.02$ \\
$ 10^6$ & $0.20$ & $0.861\times 10^{-5}$ & $0.09$ & $0.1 \times 10^{-3}$ & $0.10$ \\
$ 10^7$ & $0.25$ & $0.782\times 10^{-5}$ & $1.08$ & $0.1 \times 10^{-3}$ & $1.08$ \\
$ 10^8$ & $0.10$ & $0.943\times 10^{-5}$ & $13.1$ & $0.2 \times 10^{-3}$ & $13.1$ \\
\end{tabular}
%}
  \caption{\em Squared-exponential kernel, $\epsilon = 10^{-5}$}
 \vspace{0.5cm}
    \label{7802b}
\end{subtable}
\begin{subtable}{1.0\linewidth}
  \centering
%%%\resizebox{\columnwidth}{!}{%
\begin{tabular}{lccccccc}
  $N$ & $\rho$ & $L^2$ error & $t_{\text{form}}$ (s) & $t_{\text{solve}}$ (s) & $t_{\text{total}}$ (s) \\
  \hline
$ 10^5$ & $0.50$ & $0.805\times 10^{-3}$ & $0.02$ & $0.5 \times 10^{-4}$ & $0.02$ \\
$ 10^6$ & $0.20$ & $0.803\times 10^{-3}$ & $0.10$ & $0.7 \times 10^{-4}$ & $0.10$ \\
$ 10^7$ & $0.25$ & $0.857\times 10^{-3}$ & $1.04$ & $0.6 \times 10^{-4}$ & $1.04$ \\
$ 10^8$ & $0.10$ & $0.657\times 10^{-3}$ & $12.3$ & $0.7 \times 10^{-4}$ & $12.3$ \\
\end{tabular}
%}
\caption{\em Squared-exponential kernel, $\epsilon = 10^{-3}$}
\end{subtable}
  \caption{\em Computation times and accuracy for regression with squared-exponential kernel with $\rho \in [0.1, 0.5]$.
  The nodes and weights used are reported in Table \ref{7814} and Table \ref{t:se_nodes_weights_1e3}. }
    \label{7802c}
\end{table}

\begin{figure}[!ht]
  \centering
\begin{tikzpicture}[scale=0.9]
\centering
\begin{axis}[
    xmode=log,
    ymode=log,
    xmin=50000, xmax=200000000,
    ymin=0.002, ymax=30,
    xtick={0.0,100000,1000000,10000000,100000000},
    xlabel=$N$,
    ytick={0, 0.01, 0.1, 1.0, 10.0},
    ylabel=time (s),
    legend pos= north west
]
\addplot[line width=0.5mm, mark=*, color=blue]
    coordinates 
    {
(   100000,    0.035)
(   1000000,    0.13)
(   10000000,    1.1)
(   100000000,    12.0)
    }; 
\addplot[line width=0.5mm, mark=*, color=red, dashed]
    coordinates 
    {
(   100000,    .01)%0.005)
(   1000000,    .1)%0.06)
(   10000000,    1)%.70)
(   100000000,    10.)%8.0)
    }; 
 \legend{Algorithm \ref{a20}, $\propto N$}
\end{axis}
\end{tikzpicture}
\caption{\em Scaling times for evaluation of conditional mean for varying amounts of data
with Mat\'ern kernel. We include a plot proportional to $N$ for comparison.}
\label{2200}
\end{figure}

\section{Generalizations and Conclusions}\label{s125}
In this paper we introduce algorithms for representing and computing with Gaussian processes in $1$-dimension. In Algorithm \ref{a10}, we describe a numerical scheme for representing families of Gaussian processes as Fourier expansions of the form 
\begin{align}\label{137}
% f(x) \sim \gamma_1 \alpha_1 \cos(\xi_1 x) + \beta_1 \gamma_1 \sin(\xi_1 x) + ... + \alpha_m \gamma_m \cos(\xi_n x) + \beta_m \gamma_m \sin(\xi_m x)
f(x) \sim \sum_{i=1}^{m} \alpha_i \gamma_i \cos(2 \pi \xi_i x) + \beta_i \gamma_i \sin(2 \pi \xi_i x),
\end{align}
where for $i=1,...,m$,  the frequencies $\xi_i \in \R$ are fixed, $\gamma_i \in \R$, and $\alpha_i, \beta_i$ are iid standard normal Gaussians.
These expansions are constructed in such a way that they are valid over families of covariance kernels. Representing a GP as expansion (\ref{137}), allows the use of Algorithm \ref{a20} to perform GP regression in $O(m^3)$ operations after $O(N + m^2 \log{m})$ precomputation where $N$ is the number of data points and $m$ the size of the expansion. 

While this paper is focused on GPs in $1$-dimensions, much of the theory and numerical
machinery extends naturally to higher dimensions. 
In particular, for GPs over $\R^d$, one generalization of the $1$-dimensional Fourier 
expansion is the tensor-product expansion of the form
\begin{align}\label{high_d_gp}
f(\vct{x}) \sim \sum_{\vct{j} \in \{1,...,m\}^{d}} \alpha_{\vct{j}} e^{i 2\pi \langle \vct{\xi_j}, \vct{x} \rangle},
\end{align}
where $\vct{x} \in \R^d$, $\vct{\xi_j} = (\xi_{\vct{j}_1}, \xi_{\vct{j}_2}, ..., \xi_{\vct{j}_d})$ where 
$\xi_1,...,\xi_m$ are equispaced real-valued frequencies. 
This representation of Gaussian process distributions was used in \cite{greengard3} 
for efficient GP regression in $1$, $2$, and $3$ dimensions.  
Accompanying error analysis can be found in \cite{barnett2024}. 
As in the $1$-dimensional case, the 
weight-space linear system corresponding to \eqref{high_d_gp}
contains structure in $d$ dimensions that is amenable to fast algorithms. 
In particular, the matrix of the weight-space linear system has $d$-dimensional 
Toeplitz structure and can be applied efficiently using the non-uniform FFT, 
facilitating the use of iterative methods. 

We plan to address the higher dimensional extension of the tools of this paper in 
subsequent publications.

\appendix
\section{Detailed results of generalized quadrature}
In this section we include tables of quadrature nodes and weights generated
by Algorithm \ref{a10}.

\begin{table}[!ht]
  \centering
\resizebox{!}{8cm}{%
\begin{tabular}{cccccc}
 $i$ & nodes & weights & $i$ & nodes & weights \\
  \hline
$  1$&$    0.0960748783232733$&$    0.1933002284075283$& $ 44$&$   18.5783668916359304 $ & $    0.4901095181946152$ \\
$  2$&$    0.2949311558758212$&$    0.2064005047360611$&$ 45$&$   19.0739289334232396$&$    0.4936700008508603$ \\
$  3$&$    0.5121811831688609$&$    0.2293690006634405$&$ 46$&$   19.5535187961381212$&$    0.4806551264127251$ \\
$  4$&$    0.7553750530381257$&$    0.2574584689846101$&$ 47$&$   20.0540157049340202$&$    0.4919980518793062$ \\
$  5$&$    1.0272554775524241$&$    0.2860863008101612$&$ 48$&$   20.5303087220011804$&$    0.4886725084331118$ \\
$  6$&$    1.3267469591800769$&$    0.3123806483750251$&$ 49$&$   21.0308285956013101$&$    0.4985884540483315$ \\
$  7$&$    1.6509109783729810$&$    0.3353906370922179$&$ 50$&$   21.5033438976054896$&$    0.4525339724563319$ \\
$  8$&$    1.9964656597608781$&$    0.3552203181130367$&$ 51$&$   22.0147025165826093$&$    0.5300911113825829$ \\
$  9$&$    2.3604331716760738$&$    0.3722929241325456$&$ 52$&$   22.4786227659423616$&$    0.4703295964417554$ \\
$ 10$&$    2.7402749478447399$&$    0.3870336516793261$&$ 53$&$   22.9890428278831500$&$    0.5031381396430530$ \\
$ 11$&$    3.1338448880682450$&$    0.3998022537083804$&$ 54$&$   23.4620311240316504$&$    0.4385610519729455$ \\
$ 12$&$    3.5393213277133500$&$    0.4108896687408796$&$ 55$&$   23.9659051179655798$&$    0.5529348248583776$ \\
$ 13$&$    3.9551382871240710$&$    0.4205230182176792$&$ 56$&$   24.4374955606510014$&$    0.4879736648158969$ \\
$ 14$&$    4.3799500415428971$&$    0.4289035535544789$&$ 57$&$   24.9268582817982995$&$    0.4320150910275336$ \\
$ 15$&$    4.8125843418627330$&$    0.4361897874230700$&$ 58$&$   25.4355862458376691$&$    0.5030088994396542$ \\
$ 16$&$    5.2520252807076444$&$    0.4425587482527165$&$ 59$&$   25.8908822711706499$&$    0.4986129817647325$ \\
$ 17$&$    5.6973962780821781$&$    0.4480785412063698$&$ 60$&$   26.4135765170375585$&$    0.5698981287926346$ \\
$ 18$&$    6.1479398698978196$&$    0.4528922343761104$&$ 61$&$   26.8244665781450813$&$    0.3580010222446956$ \\
$ 19$&$    6.6029703189280022$&$    0.4571127652712647$&$ 62$&$   27.3929596371526216$&$    0.4477770521368305$ \\
$ 20$&$    7.0619657471053916$&$    0.4607536593597344$&$ 63$&$   27.8810764015625807$&$    0.6306236380217451$ \\
$ 21$&$    7.5244050646589304$&$    0.4640030341808996$&$ 64$&$   28.2872085273741511$&$    0.5684848926018280$ \\
$ 22$&$    7.9898781249339503$&$    0.4668744890692324$&$ 65$&$   28.9186103516193818$&$    0.1661641302432533$ \\
$ 23$&$    8.4580188390522402$&$    0.4693377032983527$&$ 66$&$   29.3361902609053793$&$    0.7145700506865926$ \\
$ 24$&$    8.9284697412154603$&$    0.4715695985713800$&$ 67$&$   29.7475402852299098$&$    0.3653264402765759$ \\
$ 25$&$    9.4010606000280621$&$    0.4735626358176499$&$ 68$&$   30.2919596593750207$&$    0.7987670620900847$ \\
$ 26$&$    9.8755021184860947$&$    0.4753261760015748$&$ 69$&$   31.1707837877708087$&$    0.6496969625503436$ \\
$ 27$&$   10.3515866576509996$&$    0.4768203234940871$&$ 70$&$   31.7921709344750916$&$    0.6374198048803309$ \\
$ 28$&$   10.8290813405759305$&$    0.4782250403991462$&$ 71$&$   32.0844373549515112$&$    0.4525776393523478$ \\
$ 29$&$   11.3078856905533307$&$    0.4795013762603529$&$ 72$&$   32.8290834819148998$&$    0.5792967675329964$ \\
$ 30$&$   11.7878547186769005$&$    0.4806019700535451$&$ 73$&$   33.6870452221527970$&$    1.2316989151794200$ \\
$ 31$&$   12.2690444362196605$&$    0.4816049378992831$&$ 74$&$   34.6556177080743311$&$    0.6653181217090052$ \\
$ 32$&$   12.7507923286660407$&$    0.4826033499503408$&$ 75$&$   35.5973697159031417$&$    0.7979748203948971$ \\
$ 33$&$   13.2340308273640499$&$    0.4830779751411056$&$ 76$&$   36.0608079398967192$&$    0.9871538295211217$ \\
$ 34$&$   13.7168180215943405$&$    0.4845298908539934$&$ 77$&$   37.4828497993489194$&$    1.1429690155529550$ \\
$ 35$&$   14.2023546685620694$&$    0.4833485049199899$&$ 78$&$   38.1056387575873927$&$    0.3983778654241495$ \\
$ 36$&$   14.6852339737999902$&$    0.4872979567198063$&$ 79$&$   38.4560335475650206$&$    0.7492963615598504$ \\
$ 37$&$   15.1736671929864606$&$    0.4828723127156595$&$ 80$&$   39.9230354160448471$&$    1.7911442981045280$ \\
$ 38$&$   15.6561205601530808$&$    0.4883987918117061$&$ 81$&$   41.8408663755605872$&$    1.9661520413352620$ \\
$ 39$&$   16.1465532564475396$&$    0.4852092322682524$&$ 82$&$   43.7213283268615385$&$    1.6715931870761731$ \\
$ 40$&$   16.6287678133340897$&$    0.4863731168868852$&$ 83$&$   44.2767860363890975$&$    0.2685526601061519$ \\
$ 41$&$   17.1217168205720291$&$    0.4871456130687457$&$ 84$&$   45.6320859437675992$&$    1.7725506245674350$ \\
$ 42$&$   17.6027784484422583$&$    0.4894715943022971$&$ 85$&$   47.5466594857420191$&$    1.9715488370126710$ \\
$ 43$&$   18.0975219052295202$&$    0.4843933188273860$&$ 86$&$   49.4591701554423580$&$    1.5256479816263220$ \\
\end{tabular}
}
  \caption{\em Nodes and weights for GPs defined on $[-1, 1]$ with Mat\'ern kernels with $\nu \in [0.5, 3.5]$, $\rho \in [0.1, 0.5]$ and generalized quadrature error tolerance $\epsilon = 10^{-5}$}
    \label{7813}
\end{table}

\begin{table}[!ht]
  \centering
\resizebox{0.5\textwidth}{!}{%
\begin{tabular}{ccc}
 $i$ & nodes & weights  \\
  \hline
$  1$&$    0.1229445208158333$ & $    0.2460787859722326$ \\
$  2$&$    0.3699809596646247$ & $    0.2483834500933746$ \\
$  3$&$    0.6205086829190636$ & $    0.2530883303134549$ \\
$  4$&$    0.8770552409688902$ & $    0.2604997558962125$ \\
$  5$&$    1.1425176151727869$ & $    0.2709680345476000$ \\
$  6$&$    1.4200169312069510$ & $    0.2844795049982405$ \\
$  7$&$    1.7121859512103039$ & $    0.3000950603476956$ \\
$  8$&$    2.0204572753950010$ & $    0.3164711614602873$ \\
$  9$&$    2.3450325050947400$ & $    0.3326020197065039$ \\
$ 10$&$    2.6853624988811311$ & $    0.3478765986892886$ \\
$ 11$&$    3.0404027813093690$ & $    0.3619772453611054$ \\
$ 12$&$    3.4089356692419108$ & $    0.3748597265560644$ \\
$ 13$&$    3.7897759561613382$ & $    0.3866892173786112$ \\
$ 14$&$    4.1819675933484106$ & $    0.3975832074366130$ \\
$ 15$&$    4.5848345162589501$ & $    0.4080437319776904$ \\
$ 16$&$    4.9982272332198843$ & $    0.4187013160108272$ \\
$ 17$&$    5.4216263344213846$ & $    0.4286803282747237$ \\
$ 18$&$    5.8568375371290067$ & $    0.4407898239564511$ \\
$ 19$&$    6.3049123731014509$ & $    0.4532465428424199$ \\
$ 20$&$    6.7704615862559896$ & $    0.4793970577067649$ \\
$ 21$&$    7.2856984407800462$ & $    0.5955092328616340$ \\

\end{tabular}
}
  \caption{\em Nodes and weights for GPs defined on $[-1, 1]$ with squared-exponential kernels with $\rho \in [0.1, 0.5]$ and generalized quadrature error tolerance $\epsilon = 10^{-5}$}
    \label{7814}
\end{table}

\begin{table}[!ht]
  \centering
\resizebox{0.5\textwidth}{!}{%
\begin{tabular}{ccc}
 $i$ & nodes & weights  \\
  \hline
$  1$&$    0.1422798837388160$ & $    0.2849869860136982$ \\
$  2$&$    0.4303383428161451$ & $    0.2928359950804127$ \\
$  3$&$    0.7298461801991711$ & $    0.3067700728299949$ \\
$  4$&$    1.0470392453484560$ & $    0.3273066745906839$ \\
$  5$&$    1.3851554070607901$ & $    0.3487384660693318$ \\
$  6$&$    1.7436383837560350$ & $    0.3683279469348825$ \\
$  7$&$    2.1204551041309121$ & $    0.3845946863247884$ \\
$  8$&$    2.5128009567183280$ & $    0.3990369748481500$ \\
$  9$&$    2.9171012161369898$ & $    0.4110929506127443$ \\
$ 10$&$    3.3326272282803751$ & $    0.4225819367888854$ \\
$ 11$&$    3.7585126968767120$ & $    0.4287907919621072$ \\
$ 12$&$    4.1943700876820937$ & $    0.4446683446440617$ \\
$ 13$&$    4.6430513969235117$ & $    0.4594237430500937$ \\
$ 14$&$    5.0962967069840737$ & $    0.4282067825430210$ \\
$ 15$&$    5.5883652398864117$ & $    0.5632442756521038$ \\
$ 16$&$    6.0596680913512913$ & $    0.4616691765322060$ \\
\end{tabular}
}
  \caption{\em Nodes and weights for GPs defined on $[-1, 1]$ with squared-exponential kernels with $\rho \in [0.1, 0.5]$ and generalized quadrature error tolerance $\epsilon = 10^{-3}$}
    \label{t:se_nodes_weights_1e3}
\end{table}

\newpage
\bibliographystyle{abbrv}
\bibliography{refs}

\end{document}